\documentclass{article}
\usepackage[alg]{KLMM/klmm}

\usepackage[T1]{fontenc}
\usepackage[utf8]{inputenc}
\usepackage{lmodern}
\usepackage{mathtools}
\usepackage{enumitem }

\newcommand{\Gal}{\operatorname{Gal}}

\newcommand{\softO}{\widetilde O}

\newtheorem{conjecture}[theorem]{Conjecture}
\theoremstyle{definition}

\title{Complexity of Low-Degree Skew Polynomial Multiplication over Finite Fields}
\author{Ke Ye\institution{State Key Laboratory of Mathematical Sciences,
Academy of Mathematics and Systems Science, Chinese Academy of Sciences, and
University of Chinese Academy of Sciences. 
This work is supported by the Strategic Priority
Research Program of Chinese Academy of Sciences under Grant XDA0480502 and
XDA0480503.}
    \and Yichuan Cao\instref{1}
    \and Ruichen Qiu\instref{1}
}
\date{\today}

\begin{document}
\maketitle

\begin{abstract}
In this note, we study the complexity of multiplication in skew polynomial
rings over finite fields. We prove that the product of two elements in
$\mathbb{F}_{q^n}[x;\sigma]$ of degree at most $d < n$ can be computed using $\widetilde O(d^{\omega_K-1}n)$ arithmetic operations over $\mathbb{F}_q$, where $\sigma$ is the
$q$-Frobenius automorphism. This matches the
conjectural upper bound of Caruso--Le Borgne~[ISSAC'17] and
is quasi-optimal in view of the lower bound of Chen--Ye [ISSAC'24]. The proof reduces the finite-field case to the split algebra case using the equivariant multiplication theory of Couveignes–Ezome~[J.~Algebra, 2023], and then applies existing fast algorithms.
\end{abstract}

\section{Introduction}

Skew polynomial rings are rings of non-commutative polynomials introduced by Ore \cite{Ore1933}. They appear naturally throughout computational and non-commutative algebra. For example, these non-commutative rings provide algebraic models for rank-metric and Gabidulin-type codes \cite{BoucherUlmer2009,BoucherUlmer2014,PuchingerWachterZeh2016,PuchingerWachterZeh2018}, and interact with Gr\"{o}bner bases~\cite{LaScalaLevandovskyy2013} and structured matrix multiplication \cite{HuangYeGao2024}. From the algorithmic point of view, multiplication is the basic operation on which division, factorization, interpolation, and coding-theoretic procedures depend. The first fast algorithm for skew polynomial multiplication was discovered by Giesbrecht in 1998 \cite{Giesbrecht1998}. The known fastest algorithms include the general algorithms of Caruso--Le Borgne
\cite{CarusoLeBorgne2017}, sparse-support methods of Giesbrecht--Huang--Schost
\cite{GiesbrechtHuangSchost2020}, and special low-degree algorithms of Chen--Ye
\cite{ChenYe2024}. 

Throughout the paper, we use $\mathbb{K}$ for a field, and $\mathcal{A}/\mathbb{K}$ for an \'{e}tale $\mathbb{K}$-algebra. We also assume that there is a $\mathbb{K}$-linear automorphism $\sigma$ of $\mathcal{A}$ such that $\mathcal{A}^{\sigma} = \mathbb{K}$, and the order of $\sigma$ is equal to $\dim_{\mathbb{K}} \mathcal{A}$. Such a $\mathbb{K}$-algebra is called a $\langle \sigma \rangle$-Galois algebra. Two typical examples are $(\mathbb{K}^n, \tau)$ and $(\mathbb{L},\sigma)$, where $\mathbb{K}^n$ is the split $\mathbb{K}$-algebra and $\tau$ is the cyclic left shift operator sending $(a_1,\dots, a_n)$ to $(a_2,\dots, a_{n},a_1)$, $\mathbb{L}$ is a cyclic extension of $\mathbb{K}$ and $\sigma \in \Gal(\mathbb{L}/\mathbb{K})$ is a generator. 

The underlying additive group of the skew polynomial ring $\mathcal{A}[x;\sigma]$ is the ordinary polynomial ring $\mathcal{A}[x]$. Given $A, B \in \mathcal{A}[x;\sigma]$, we write $f = \sum_i a_i x^i$, $g = \sum_j b_j x^j$ for some $a_i, b_j \in \mathcal{A}$. Then the product $fg \in \mathcal{A}[x;\sigma]$ is defined as $fg = \sum_{k} c_k x^k$, where 
\begin{equation}\label{eq:coeff-formula}
        c_k \coloneqq \sum_{i+j=k} a_i\,\sigma^i(b_j).
\end{equation}
Given a positive integer $d$, we denote by $\mathcal{A}[x;\sigma]_{\le d}$ the subspace of $\mathcal{A}[x;\sigma]$ consisting of all skew polynomials of degree at most $d$. We consider the map
\[
\mu_d: \mathcal{A}[x;\sigma]_{\le d} \times \mathcal{A}[x;\sigma]_{\le d} \to \mathcal{A}[x;\sigma]_{\le 2d},\qquad (f,g) \mapsto fg.
\]
Let $C_\mathbb{K}(\mu_d)$ be the computational complexity of $\mu_d$. Concerning the value of $C_\mathbb{K}(\mu_d)$, we have the following conjecture.

\begin{conjecture}\cite{CarusoLeBorgne2017}\label{conj:1}
For any $\langle \sigma \rangle$-Galois algebra $\mathcal{A}/\mathbb{K}$ and positive integer $d$, we have $C_\mathbb{K}(\mu_d) = \softO\bigl(dn\min(d,n)^{\omega_{\mathbb{K}}-2}\bigr)$ where $n = \dim_{\mathbb{K}} \mathcal{A}$ and $\omega_\mathbb{K}$ is the exponential of the complexity of the matrix multiplication over $\mathbb{K}$.
\end{conjecture}

According to \cite{CarusoLeBorgne2017}, Conjecture~\ref{conj:1} is known to be true for $d \ge n$. Thus, it suffices to assume $d < n$. Moreover, it is proved in \cite{ChenYe2024} that $C_\mathbb{K}(\mu_d)$ is bounded below by $\Omega\bigl(dn\min(d,n)^{\omega_{\mathbb{K}}-2}\bigr)$, showing that any algorithm for $\mu_d$ achieving the conjectured upper bound is quasi-optimal. The goal of this note is to prove Conjecture~\ref{conj:1} for the case where $\mathcal{A}/\mathbb{K} = \mathbb{F}_{q^n}/\mathbb{F}_q$. 

\begin{theorem}\label{thm:finite-field}
Let $q$ be a prime power and let $n$ be a positive integer. If $\sigma$ is the Frobenius generator of $\Gal(\mathbb{F}_{q^n}/\mathbb{F}_q)$, then for any $0 < d < n$, we have
\[
C_\mathbb{K}(\mu_d) = \softO(d^{\omega_{\mathbb{K}}-1}n).
\]
\end{theorem}

The main results presented in this paper are obtained through an interaction between the authors and an artificial intelligence agent system, \textit{MechMath Agent Team (MMAT)}~\cite{MMAT}.
The authors assume full responsibility for the paper’s content.

\section{Preliminaries}\label{sec:prelim}

In this section, we record some ingredients that are required for the proof of Theorem~\ref{thm:finite-field}.

\begin{theorem}\label{src:CLB}
Let $\mathcal{A}/\mathbb{K}$ be an $n$-dimensional $\langle \sigma \rangle$-Galois algebra over $\mathbb{K}$. Given $f,g \in \mathcal{A}[x;\sigma]$, we have  
\begin{enumerate}[label=(\alph*),leftmargin=2.2em]
\item If $\deg f + \deg g \leq D < n$, then $fg$ can be computed in $\softO(D^{\omega_{\mathbb{K}}-2}n^2)$ operations over $\mathbb{K}$. \label{src:CLB:item1}
\item For any $d \geq n$, we have $C_{\mathbb{K}}(\mu_d) = \softO(d n^{\omega_{\mathbb{K}}-1})$. \label{src:CLB:item2}
\end{enumerate}
\end{theorem}

\begin{theorem}\cite{ChenYe2024}\label{src:CY}
Let $\mathcal{A}/\mathbb{K}$ be the $n$-dimensional split $\mathbb{K}$-algebra with the cyclic left shift $\tau$ sending $(a_1,\ldots,a_r)$ to $(a_2,\ldots,a_r,a_1)$. For any $d < n/3$, we have $C_{\mathbb{K}}(\mu_d) = \softO(d^{\omega_{\mathbb{K}}-1}n)$.
\end{theorem}

\begin{theorem}\cite{CouveignesEzome2022}\label{src:CE}
For any prime power $q$ and positive integer $n$, the multiplication map
$\mathbb{F}_{q^n} \times \mathbb{F}_{q^n} \to \mathbb{F}_{q^n}$ has symmetric $G$-equivariant complexity $\nu_q^{\mathrm{sym}}(n) \leq C \lceil \log_q n\rceil$
for some absolute constant $C > 0$. Here $G \coloneqq \Gal(\mathbb{F}_{q^n}/\mathbb{F}_{q})$.
\end{theorem}
Let $\mathcal{E}$ be the $n$-dimensional split $\mathbb{F}_q$-algebra. By definition, $\nu_q^{\mathrm{sym}}(n) = s$ is equivalent to saying that there are $\mathbb{K}[G]$-linear maps
\[
T:\mathbb{F}_{q^n} \longrightarrow \mathcal E^s,\qquad B:\mathcal E^s\longrightarrow \mathbb{F}_{q^n}
\]
such that
\begin{equation}\label{eq:CE-factorization}
        uv=B\bigl(T(u)\diamond_s
        T(v)\bigr), \qquad u,v\in \mathbb{L},
\end{equation}
where $\diamond_s$ means componentwise multiplication in each of the $s$ copies of
$\mathcal E$. 

\section{Proof of Theorem~\ref{thm:finite-field}}\label{sec:finite-main}

In this section, we denote $\mathbb{K} \coloneqq \mathbb F_q$ and $\mathbb{L} \coloneqq \mathbb F_{q^n}$. Then $G \coloneqq \Gal(\mathbb{L}/\mathbb{K}) = \langle \sigma \rangle$, where $\sigma$ is the Frobenius automorphism. We also denote by $\mathcal{E}$ the $n$-dimensional split $\mathbb{K}$-algebra. Let $T$ be the map defined in \eqref{eq:CE-factorization} and let $\tau$ be the automorphism of $\mathcal{E}$ such that 
\begin{equation}\label{eq:orientation}
        T(\sigma z)=\tau T(z),\qquad z\in \mathbb{L}.
\end{equation}
It is straightforward to verify that $\tau$ is the inverse of the cyclic left shift operator. Clearly, Theorem~\ref{src:CY} still holds for this choice of $\tau$.

Given $f =\sum_{i=0}^{d} a_i x^i$ and $g =\sum_{j=0}^{d} b_j x^j$ in $\mathbb{L}[x;\sigma]$, we write $T(a_i)=(F_{i,t})_{t=1}^s$ and $T(b_j)=(G_{j,t})_{t=1}^s$ with $F_{i,t}, G_{j,t}\in\mathcal E$. For each $1\le t\le s$, we define skew polynomials $P_t=\sum_{i=0}^d F_{i,t}x^i$ and $Q_t =\sum_{j=0}^d G_{j,t}x^j$ in $\mathcal E[x;\tau]$, where $xe=\tau(e)x$ for $e\in\mathcal E$. We denote by $S_{k,t}$ the
coefficient of $x^k$ in the product $P_tQ_t$. 

\begin{lemma}\label{lem:coeff-bridge}
For each $0\leq k\leq 2d$, the coefficient $c_k$ of
$x^k$ in $fg$ is
\begin{equation}\label{eq:bridge-output}
        c_k=B \bigl((S_{k,t})_{t=1}^s\bigr).
\end{equation}
\end{lemma}

\begin{proof}
The $k$-th coefficient $c_k$ of $AB$ is given in \eqref{eq:coeff-formula}. By \eqref{eq:CE-factorization}, we have $ a_i\sigma^i(b_j) = B\bigl(T(a_i)\diamond_s T(\sigma^i b_j)\bigr)$.
According to \eqref{eq:orientation}, we deduce that $T(\sigma^i b_j)=\tau^i T(b_j) =\bigl(\tau^i G_{j,t}\bigr)_{t=1}^s$, from which we conclude that 
\[
        T(a_i)\diamond_s T(\sigma^i b_j)
        =\bigl(F_{i,t}\diamond \tau^i(G_{j,t})\bigr)_{t=1}^s.
\]
Since $B$ is $\mathbb{K}$-linear, this implies 
\[
c_k =\sum_{i+j=k}B\bigl((F_{i,t}\diamond \tau^i(G_{j,t}))_{t=1}^s\bigr) = B \Big( \Big(\sum_{i+j=k}F_{i,t}\diamond \tau^i(G_{j,t})\Big)_{t=1}^s\Big) = B( (S_{k,t})^s_{t=1}).\qedhere
\] 
\end{proof}

Now we are ready to prove Theorem~\ref{thm:finite-field}.

\begin{proof}[Proof of Theorem~\ref{thm:finite-field}]
Denote $s \coloneqq \nu_q^{\mathrm{sym}}(n)$. By Theorem~\ref{src:CE}, we have $s=\softO(1)$. The rest of the proof is divided into three cases.

\smallskip
\noindent\textbf{Case I: $1\le d<n/3$.}
Suppose that $T$ and $B$ are maps defined in \eqref{eq:CE-factorization}. By choosing normal $\mathbb{K}[G]$-module bases, both $\mathbb{F}_{q^n}$ and $\mathcal{E}$ are identified with the regular module $\mathbb{K}[G]$. If we write $T = (T_1,\dots, T_s)$, then each component
$T_t: \mathbb{F}_{q^n} \to \mathcal{E}$ is identified with the convolution by some fixed element of $\mathbb{K}[G]$. On the other hand, we may write $B = \sum_{t=1}^s B_t$ where $B_t: \mathcal{E} \to \mathbb{F}_{q^n}$ is the composition of $B$ with the natural inclusion $\mathcal{E} \hookrightarrow \mathcal{E}^s$ into the $t$-th component. We may identify each $B_t$ with the convolution by some fixed element of $\mathbb{K}[G]$. Consequently, both $T$ and $B$ can be computed by $\softO(n)$ operations as $s=\softO(1)$.

We count the number of operations required in Lemma~\ref{lem:coeff-bridge}. Firstly, applying $T$ to all $2(d+1)$ input coefficients. This costs $\softO(dn)$ operations over $\mathbb{K}$. Note that $\omega_{\mathbb{K}}\ge 2$, we have $\softO(dn) \le \softO(d^{\omega_\mathbb{K} - 1}n)$. For each $1\le t\le s$, we invoke Theorem~\ref{src:CY} to compute $P_tQ_t$ by $\softO(d^{\omega_{\mathbb{K}}-1}n)$ operations. Applying $B$ to the $2d+1$ output coefficient vectors only costs $\softO(dn)$ operations. Thus, the total cost is $\softO(d^{\omega_{\mathbb{K}}-1}n)$.

\smallskip
\noindent\textbf{Case II: $n/3\le d<n/2$.}
Set $D=2d$. Then $D < n$ and Theorem \ref{src:CLB}\ref{src:CLB:item1} implies 
\[
C_{\mathbb{K}}(\mu_d) = \softO(D^{\omega_{\mathbb{K}}-2}n^2)=\softO(d^{\omega_{\mathbb{K}}-2}n^2) = \softO( d^{\omega_{\mathbb{K}}-1}n ),
\]
where the last equality follows from the assumption that $d\ge n/3$.

\smallskip
\noindent\textbf{Case III: $n/2\le d<n$.}
We view $\mathbb{L}[x;\sigma]_{\le d}$ as a subspace of $\mathbb{L}[x;\sigma]_{\le n}$. Since $n \le 2d$, Theorem \ref{src:CLB}\ref{src:CLB:item2} yields
\[
C_{\mathbb{K}}(\mu_d) \le C_{\mathbb{K}}(\mu_n) = \softO(n^{\omega_{\mathbb{K}}}) = \softO( d^{\omega_{\mathbb{K}}-1} n). \qedhere
\]
\end{proof}

\end{document}